\documentclass[oldstyle]{mtt}

\usepackage{array}
\usepackage[unicode]{hyperref}
\addtocounter{page}{46}

\def\mP{P}
\def\bbR{\mathbb{R}}
\def\ds{\displaystyle}
\def\mct{\mathcal{R}}

\newcommand{\fns}[1]{{\footnotesize #1}}
\newcommand{\hf}[1]{\hfil #1 \hfil}
\newcommand{\ru}[1]{\rule{0pt}{#1pt}}

\udk{531.38+517.938.5}

\title[Аналитическая классификация равномерных вращений гиростата Ковалевской\,--\,Яхья]{АНАЛИТИЧЕСКАЯ КЛАССИФИКАЦИЯ РАВНОМЕРНЫХ\\ВРАЩЕНИЙ ГИРОСТАТА КОВАЛЕВСКОЙ\,--\,ЯХЬЯ
}

\author{М.П.~Харламов}

\thanks{\hspace{-5mm}{Работа выполнена при финансовой поддержке РФФИ (грант № 10-01-00043).}}

\mtt{year=2012,number=42,received=11.08.12}

\address{Волгоградский филиал РАНХиГС, Россия}
\email{mharlamov@vags.ru}

\begin{document}

\maketitle

\begin{abstract}[ru]
Представлено полное исследование множества равномерных вращений гиростата в случае интегрируемости Ковалевской\,--\,Яхья.
Введено понятие классов эквивалентности относительно определяющих параметров, построено разделяющее множество. Для каждого класса вычислен тип особенности как тип неподвижной точки в приведенной системе, получен характер устойчивости, указана структура локального слоения Лиувилля.
\end{abstract}

\section*{Введение}\label{sec1}
Движение гиростата Ковалевской--Яхья описывается уравнениями
\begin{equation}\label{eq:1}
\begin{array}{lll}
2\dot\omega _1   = \omega _2 (\omega _3- \lambda)  , &
2\dot\omega _2 =  - \omega _1 (\omega _3-\lambda)  - \alpha _3 , &
\dot\omega _3   = \alpha _2, \\
\dot\alpha _1   = \alpha _2 \omega _3  - \alpha _3 \omega _2 , &
\dot\alpha _2   = \alpha _3 \omega _1  - \alpha _1 \omega _3 , &
\dot\alpha_3   = \alpha_1 \omega_2  - \alpha_2 \omega_1.
\end{array}
\end{equation}
Здесь $\bs \omega$ -- угловая скорость, $\bs \alpha$ -- вектор силового поля, $\lambda > 0$ -- гиростатический момент. Фазовое пространство $\mP^5=\bbR^3_{\omega}{\times}S^2_{\alpha}$ определено в $\bbR^6$ геометрическим интегралом $\Gamma={\bs \alpha}^2$, постоянная которого принимается равной единице. Первые интегралы системы таковы
\begin{equation}\notag
\begin{array}{l}
H = \omega _1^2  + \omega _2^2 + \ds{\frac{1}{2}}\omega _3^2 -
\alpha _1, \qquad
L = \omega _1 \alpha _1  + \omega _2 \alpha _2  +\frac{1}{2} (\omega _3+\lambda) \alpha _3, \\
K=(\omega_1^2-\omega^2_2+\alpha_1)^2+(2\omega_1\omega_2+\alpha_2)^2 + 2\lambda[(\omega_3-\lambda)(\omega_1^2+\omega^2_2)+2\omega_1 \alpha_3].
\end{array}
\end{equation}
Интеграл $K$ указан в работе \cite{Yeh1}. Система \eqref{eq:1} имеет гамильтонову форму $\dot x =\{H,x\}$ относительно скобки Ли\,--\,Пуассона на пространстве $\bbR^6({\bs \omega},{\bs \alpha})$. Функции $\Gamma$ и $L$ являются аннуляторами скобки, поэтому на любом их совместном уровне $\mP^4_\ell=\{\Gamma=1, L=\ell\}$ возникает гамильтонова система с двумя степенями свободы (приведенная система).

Аналитическому и топологическому исследованию системы \eqref{eq:1} посвящено много работ. Основные из них указаны, например, в библиографии статьи \cite{mtt40}. Однако до сих пор нет строгой и полной классификации такого важного типа движений гиростата Ковалевской\,--\,Яхья, как равномерные вращения.

Известно, что в поле силы тяжести равномерные вращения (движения с постоянной угловой скоростью) возможны только вокруг вертикали, поэтому соответствующие значения пары $({\bs \omega},{\bs \alpha})$ являются неподвижными точками системы \eqref{eq:1} и ``положениями равновесия'' в приведенных системах. В силу этого равномерные вращения называют также относительными равновесиями. Ниже относительным равновесием называется именно неподвижная точка системы \eqref{eq:1}, а под характером устойчивости равномерного вращения понимается характер устойчивости соответствующей неподвижной точки в содержащей ее приведенной системе на $\mP^4_\ell$.

Равномерными вращениями гиростатов различных конфигураций занимались многие авторы. Вероятно, наиболее полное изложение истории вопроса можно найти в \cite{Gorr}. Некоторая информация об устойчивости относительных равновесий в случае Ковалевской\,--\,Яхья может быть получена из работы \cite{Gash1}, где, в частности, приведены условия существования движений, асимптотических к семейству периодических траекторий, включающих и все относительные равновесия как точки бифуркации. Вопрос существования асимптотических движений для этого же семейства изучался в работе \cite{Malah}, но полученные в \cite{Malah} условия сразу же следуют из результатов \cite{Gash1}.
Однако в задаче классификации неподвижных точек интегрируемой системы информация о наличии или отсутствии асимптотических движений является недостаточной для строгого описания характера устойчивости. Гораздо более точными характеристиками являются тип неподвижной точки и ее круговая молекула. Тип точки полностью определяет детальное свойство устойчивости (количество переменных, по которым неподвижная точка устойчива или неустойчива, и соответствующие направления в фазовом пространстве), а круговая молекула описывает топологию системы в окрестности неподвижной точки с точностью до изоморфизма слоения на интегральные многообразия.

Цель настоящей работы -- в терминах подходящих параметров определить понятие класса эквивалентности относительных равновесий гиростата Ковалевской\,--\,Яхья, указать все разделяющие значения параметров, для каждого класса указать тип относительного равновесия как особой точки интегрируемой гамильтоновой системы, характер устойчивости, топологию совместного уровня первых интегралов, содержащего относительное равновесие, и возникающее в его окрестности слоение интегральных многообразий.

\section{Множество относительных равновесий}

Множество точек $({\bs \omega},{\bs \alpha})$ в $\mP^5$, отвечающих равномерным вращениям, т.е. множество относительных равновесий всех приведенных систем, обозначим через $\mct$.  Локально множество $\mct$ параметризовано значением $\ell$ интеграла площадей (для приведенных систем относительные равновесия -- это критические точки ранга $0$). В то же время, на каждом $\mP^4_\ell$ имеется конечное число точек $\mct$, существенно меняющееся с изменением $\ell$, поэтому аналитические выражения в зависимости от $\ell$ плохо подходят для исследований.

Известно, хотя специально это нигде не отмечено, что все равномерные вращения в случае Ковалевской\,--\,Яхья содержатся в семействе траекторий, отвечающем решению, найденному П.В.\,Харламовым в работе \cite{PVLect}. Это решение обобщает на гиростат с осевой динамической симметрией и произвольным отношением осевого и экваториального моментов инерции случай интегрируемости Бобылева--Стеклова. Если первая ось инерции, как и в уравнениях~\eqref{eq:1}, выбрана в направлении центра масс, а последний лежит в экваториальной плоскости, то соответствующее фазовое пространство задано системой инвариантных соотношений $\{\omega_2 =0, \, \dot \omega_2 =0\}$. В принятых здесь обозначениях явные формулы решения \cite[\S\,5.4]{PVLect} имеют вид
\begin{equation}\label{eq:2}
\begin{array}{lll}
  \ds{\omega_1=p_0,} &  \ds{\omega_2=0,} & \ds{\omega_3 = r,}\\
  \ds{\alpha_1=p_0^2+\frac{1}{2} r^2-h,} &
  \ds{\alpha_2= f(r), }&
  \ds{\alpha_3 = -p_0 (r-\lambda),}\\
\end{array}
\end{equation}
где $h$ -- произвольная постоянная интеграла энергии, $p_0$ -- свободный параметр, $r$ -- независимая переменная, подчиненная уравнению ${\dot r} = f(r)$, и
$$
\ds f^2(r) = -\frac{1}{4}r^4-(2p_0^2-h)r^2+2 \lambda p_0^2 r+1-(p_0^2-h)^2-p_0^2\lambda^2.
$$
Относительное равновесие получим, если величина $r$ постоянна и равна кратному корню многочлена $f^2(r)$. Тогда
\begin{equation}\label{eq:3}
    \begin{array}{c}
      \omega_1 = \pm V, \quad \omega_2 = 0, \quad \omega_3=r ,\\
      \alpha_1 = \displaystyle{\frac{1}{2}\left[r(r-\lambda)- d\right]}, \quad \alpha_2 = 0, \quad \alpha_3 = \mp (r-\lambda) V,
    \end{array}
\end{equation}
где
\begin{equation}\notag
    d^2 = r^2(r-\lambda)^2+4, \qquad V=\sqrt{\displaystyle{\frac{r}{2}\left[-r+\frac{1}{r-\lambda}d\right]}} \geqslant 0.
\end{equation}
Знаки $\omega_1, \alpha_3$ согласованы $($оба верхние или оба нижние$)$. Изучая кривые пересечения бифуркационных поверхностей, И.Н.\,Гашененко \cite{Gash4} указал область изменения константы $r$, рассматриваемой как единственный свободный параметр, правило выбора знака $d$, а также явные выражения первых интегралов. Эти результаты непосредственно применимы к относительным равновесиям.

\begin{proposition}[\cite{Gash4}]\label{propos2}
{\it Для относительных равновесий параметр $r$ пробегает множество}
\begin{equation}\label{eq:4}
    r \in (-\infty,0] \cup [0,\lambda) \cup (\lambda,+\infty),
\end{equation}
{\it знак $d$ совпадает со знаком $r (r-\lambda)$ при $r \ne 0$ и произволен при $r=0$, а значения первых интегралов таковы}
\begin{equation}\notag
\begin{array}{c}
    \ell = \mp \ds{\frac{1}{2}[\lambda(r-\lambda)+d]} V,\quad
    h= -\ds{\frac{1}{2}r(r-\lambda)+\frac{2r-\lambda}{2(r-\lambda)}d}, \\
    k  = \ds{\frac{\lambda}{4(r-\lambda)^2} [r(r-\lambda)-d][r(r-\lambda)(4 r-3\lambda)-\lambda d]}.
\end{array}
\end{equation}
\end{proposition}

Из \eqref{eq:3}, \eqref{eq:4} сразу же следует, что множество $\mct$ имеет ровно четыре связных компоненты, гомеоморфных вещественной прямой при фиксированном $\lambda>0$. В соответствии с промежутками изменения $r$ введем обозначения для подмножеств в $\mct$, определяемых формулами \eqref{eq:3}:
\begin{equation}\label{eq:5}
    \begin{array}{llll}
      \mct_1: & r \in [0,\lambda), & d<0, & \ds{\lim_{r\to +0}} d = -2, \\
      \mct_2: & r \in (-\infty,0] , & d > 0, & \ds{\lim_{r\to -0}} d = 2, \\
      \mct_3: & r \in (\lambda, +\infty), & d > 0.&
    \end{array}
\end{equation}
Первые два множества связны, последнее состоит из двух компонент, отличающихся знаком $\omega_1$. В $\mct_1$ и $\mct_2$ каждому значению $r\ne 0$ отвечают ровно две точки. В $\mct_1$ значению $r=0$ отвечает точка ${\bs \omega}=0$, ${\bs \alpha}=\{1,0,0\}$ нижнего (абсолютного) положения равновесия тела. Для дальнейшего обозначим ее в соответствии со стремлением $r$ к нулю справа через $c_+$. В $\mct_2$ нулевое значение $r$ приводит к точке ${\bs \omega}=0$, ${\bs \alpha}=\{-1,0,0\}$ верхнего (абсолютного) положения равновесия. Обозначим ее через $c_-$. На множествах \eqref{eq:5} определена очевидная симметрия
${\bs \tau}: (\omega_1,\alpha_3) \mapsto (-\omega_1,-\alpha_3)$,
которая меняет знак постоянной площадей $\ell$, связные множества $\mct_1$, $\mct_2$ переводит в себя, а в множестве $\mct_3$ меняет местами связные компоненты. Устройство семейства множеств $\mct(\lambda)$ проиллюстрировано на рис.~\ref{fig1}.

\begin{figure}[h!]
\centering
\includegraphics[width=80mm,keepaspectratio]{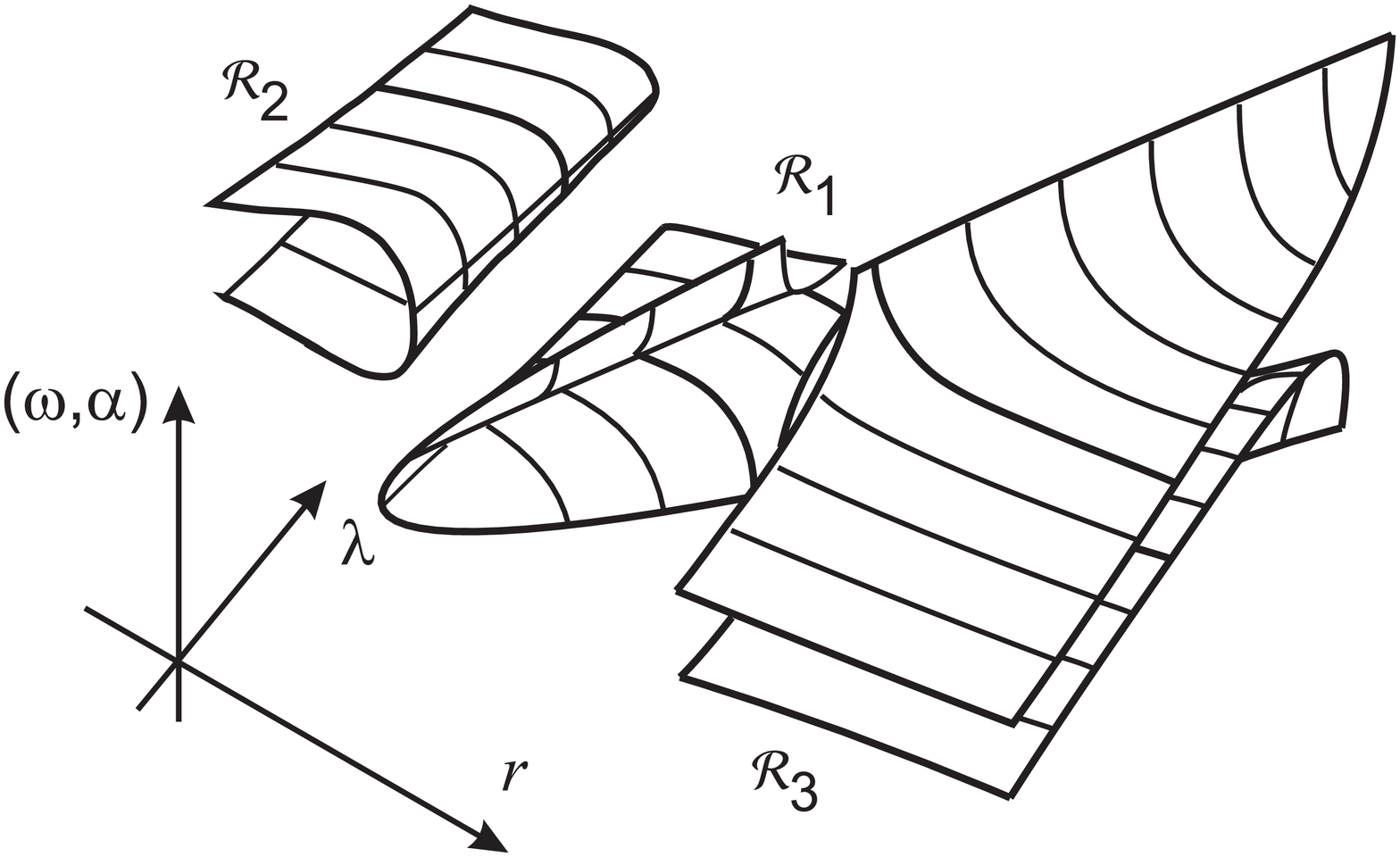}\\
\caption{Поверхности относительных равновесий.}\label{fig1}
\end{figure}

Параметр $\lambda$ обычно трактуется как физический. Однако гиростат -- это система с четырьмя степенями свободы (тело плюс ротор), в которой $\lambda$ есть константа циклического интеграла, а уравнения Эйлера\,--\,Пуассона описывают систему, полученную понижением порядка. Поэтому все исследования каких-либо свойств системы относительно интегральных постоянных естественно проводить в {\it расширенном} пространстве параметров, включающем ось значений $\lambda$. Введем следующее понятие. Если $\{S=S(\lambda)\}$ -- семейство зависящих от параметра $\lambda$ множеств в $\mP^5$, определим {\it расширенное множество}
\begin{equation}\notag
    \Lambda({S}) = \bigcup_{\lambda} S(\lambda)\times \{\lambda\} \subset \mP^5 \times \bbR=\{({\bs \omega},{\bs \alpha},\lambda): {\bs \omega}\in \bbR^3, {\bs \alpha} \in S^2, \lambda \in \bbR\}.
\end{equation}
Считая $\lambda>0$, случай $\lambda=0$ рассматриваем лишь как предельную возможность там, где это явно оговорено. Как следует из предложения~\ref{propos2}, множество $\Lambda({\mct})$ в части $\lambda>0$ гомеоморфно четырем плоскостям и дважды накрывает область $\mathcal{D}=\{(r,\lambda)\in \bbR^2: \lambda >0, r \ne \lambda \}.$ Фактически это накрытие и показано на рис.~\ref{fig1}.

Условимся образы множеств $\Lambda({\mct_i})$ при отображении на плоскость параметров $(r,\lambda)$ обозначать через $\delta_i$ $(i=1,2,3)$:
\begin{equation}\notag
    \begin{array}{ll}
      \delta_1: \{(r,\lambda): 0 \leqslant r < \lambda, \; \lambda >0 \}, \\
      \delta_2: \{(r,\lambda): r \leqslant 0, \; \lambda >0 \}, \\
      \delta_3: \{(r,\lambda): r > \lambda, \; \lambda >0 \}.
    \end{array}
\end{equation}
При дополнительной детализации будем снабжать подмножества этих множеств двойными индексами.

\section{Классы эквивалентности и типы относительных равновесий}

Напомним некоторые понятия, связанные с классификацией неподвижных точек в интегрируемой системе с двумя степенями свободы \cite{BolFom}.
Все объекты предполагаются аналитическими. Система, заданная гамильтонианом $F$, обозначается $\mathop{\rm sgrad}\nolimits F$ и в локальных координатах имеет вид
\begin{equation}\label{eq:6}
    \dot x = \Omega(x) \nabla F(x), \qquad x \in \bbR^4.
\end{equation}
Здесь $\Omega$ -- невырожденная кососимметрическая $4{\times}4$-матрица, которая определяется симплектической структурой или, как в случае уравнений Эйлера\,--\,Пуассона, вырожденной скобкой Пуассона на некотором объемлющем пространстве.

Пусть $x_0$ -- неподвижная точка системы, т.е. $\nabla F(x_0)=0$. Тогда линеаризованная система в точке $x_0$ задана матрицей
\begin{equation}\notag
    A_F = \Omega (x_0) D_F (x_0) \qquad \left(D_F = \| \frac{\partial^2 F}{\partial x_i \partial x_j}\| \right).
\end{equation}
Хорошо известно, что характеристическое уравнение матрицы $A_F$ биквадратное, поэтому его корни либо разбиты на две пары, каждая из которых имеет вид $\pm a$ или $\pm b \, \mathrm{i}\,$, либо составляют четверку $\pm a \pm b\, \mathrm{i}\,$. Пусть $\det A_F \ne 0$. Тогда в первом случае $\bbR^4$ есть прямая сумма двух инвариантных плоскостей, ограничения на которые линеаризованной системы имеют начало координат седлом или центром. Во втором случае начало координат в $\bbR^4$ является фокусной особенностью. В соответствии с этим говорят, что точка $x_0$ имеет тип ``центр-центр'' в случае двух пар мнимых корней, ``центр-седло'' (или ``седло-центр'', если важно подчеркнуть порядок следования корней) в случае пары мнимых и пары вещественных корней, ``седло-седло'' в случае двух пар вещественных корней, ``фокус-фокус'' в случае четверки комплексных корней.
В интегрируемой системе тип точки полностью определяет характер устойчивости: точка типа ``центр-центр'' устойчива по всем переменным, точка типа ``центр-седло'' по двум переменным устойчива и по двум -- неустойчива, остальные -- неустойчивы по всем переменным.

\begin{remark}\label{remm1}
{\it Пусть система \eqref{eq:6} изначально задана на некотором объемлющем пространстве {\rm (}как, например, $\bbR^6({\bs \omega},{\bs \alpha})$ в случае системы Эйлера\,--\,Пуассона{\rm )}, а фактическое фазовое пространство {\rm (}в нашей задаче -- $\mP^4_\ell${\rm )} является совместным уровнем первых интегралов -- аннуляторов скобки Пуассона {\rm (}в нашем случае это интегралы $\Gamma$ и $L${\rm )}. Пусть $\Phi$ -- один из таких интегралов. Тогда $A_F \nabla \Phi=0$, и этот интеграл порождает нулевое собственное значение матрицы $A_F$. Поэтому при вычислении собственных чисел оператора $A_F$ {\rm (}в отличие от вычисления собственных чисел ограничения оператора $D_F$ на подмногообразия совместного уровня{\rm )} интегралы такого рода учитывать не нужно, а следует просто отбросить соответствующее количество нулевых корней характеристического многочлена, которые заведомо существуют. Таким образом, нет необходимости переходить к каким-либо специальным координатам, отвечающим приведенным системам, а все вычисления легко проделать в исходных переменных {\rm (}в нашем случае -- в переменных Эйлера\,--\,Пуассона{\rm )}. Говоря далее о собственных значениях операторов вида $A_F$, вычисленных в переменных Эйлера\,--\,Пуассона, всегда имеем в виду, что два нулевых значения уже отброшены.
}
\end{remark}

Отметим еще одно преимущество перехода от $D_F$ к $A_F$ -- характеристическое уравнение $A_F$ зависит не от самой точки $x_0$, а от значений первых интегралов в этой точке.

Пусть $G$ -- первый интеграл, независимый с $F$ почти всюду в окрестности точки $x_0$ и такой, что $x_0$~-- неподвижная точка для $\mathop{\rm sgrad}\nolimits G$ (последнее заведомо выполнено, если $\det D_F(x_0) \ne 0$).

\begin{definition}[~\cite{BolFom}]\label{def3}
Неподвижная точка $x_0$ называется невырожденной, если:

$1^\circ)$ матрицы $A_F$ и $A_G$ линейно независимы;

$2^\circ)$ существует линейная комбинация матриц $A_F$ и $A_G$, у которой все четыре собственных значения различны.
\end{definition}

Невырожденность позволяет описать топологию интегральных многообразий в окрестности точки $x_0$ в виде почти прямого произведения атомов одномерных систем. При этом не запрещается ``странная'' динамика, в которой все точки одного из сомножителей могут оказаться неподвижными. Дополнительное условие  $\det A_F\ne 0$ гарантировало бы, что неподвижная точка изолирована в четырехмерном фазовом пространстве. Ниже показано, что в рассматриваемой задаче в невырожденных точках это условие выполняется. Для краткости вырожденную точку $x_0$ назовем \emph{сильно} вырожденной, если для нее нарушается условие $1^\circ$ определения~\ref{def3}, и \emph{слабо} вырожденной в противном случае. Эти термины не являются стандартными, а из сильной вырожденности не следует слабая.

Для уравнений Эйлера\,--\,Пуассона 2-форма, индуцированная на $\mP^5$ симплектической структурой многообразия $TSO(3)$, заведомо вырождена. Введенные понятия необходимо рассматривать с точки зрения систем на $\mP^4_\ell$. Но, как отмечалось, в явном переходе к этим системам, которые сделали бы вычисления необозримыми, нет необходимости. Корректно определены скобки Ли\,--\,Пуассона на $\bbR^6({\bs \omega},{\bs \alpha})$, поэтому поле $\mathop{\rm sgrad}\nolimits F$, сопоставленное функции $F$, определено уравнениями
\begin{equation}\label{eq:7}
\dot {\mathbf{M}}=\mathbf{M} \times \ds{\frac{\partial F}{\partial \mathbf{M}}} + {\bs \alpha} \times \ds{\frac{\partial F}{\partial {\bs \alpha}}},\quad
\dot {{\bs \alpha}}= {\bs \alpha} \times \ds{\frac{\partial F}{\partial \mathbf{M}}}.
\end{equation}
Здесь $\mathbf{M} = \mathbf{I} {\bs \omega} +{\bs \lambda}$, $\mathbf{I}=\mathop{\rm diag}\nolimits\{2,2,1\}$, ${\bs \lambda} =(0,0,\lambda)$.

Вернемся к рассматриваемой системе.

\begin{definition}\label{def7}
Будем говорить, что точки $\xi_1,\xi_2 \in \Lambda({\mct})$ принадлежат к одному классу, если существует непрерывный путь~в\, $\Lambda({\mct})$, соединяющий $\xi_1$ с $\xi_2$ или $\xi_1$ с ${\bs \tau}(\xi_2)$, вдоль которого тип точки не меняется.
\end{definition}

Пусть $(r,\lambda) \in \mathcal{D}$. При $r\ne 0$ и $\mathop{\rm sgn}\nolimits d=\mathop{\rm sgn}\nolimits r(r-\lambda)$ обозначим через $\xi_\pm(r,\lambda)$ две точки \eqref{eq:3} в $\mP^5$. В соответствии с принятыми ранее обозначениями имеем
\begin{equation}\notag
\lim _{r \to +0} \xi_\pm(r,\lambda) = c_+ \in \mct_1, \qquad \lim _{r \to -0} \xi_\pm(r,\lambda) = c_- \in \mct_2.
\end{equation}

\begin{definition}\label{def8}
Точку $(r,\lambda)$ назовем разделяющей, если в любой ее окрестности найдутся образы точек из $\Lambda({\mct})$ разных классов.
\end{definition}

Обозначим через $\overline{\pi}$ луч запрещенных точек $\{{r=\lambda,\, \lambda > 0}\}$, не включенный в $\mathcal{D}$. Он является разделяющим, если в качестве области изменения параметров рассматривать всю полуплоскость $\{\lambda \geqslant 0\}$, т.е. замыкание множества $\mathcal{D}$.
Поскольку $\xi_1\in \mct_1$ и $\xi_2\in \mct_2$ не могут принадлежать одному классу, точки вида $(0,\lambda)$ всегда являются разделяющими. Обозначим полуось $\{{r=0,\,\lambda \geqslant 0}\}$ через $\pi_0$. Подчеркнем, что точки луча $\pi_0$ разделяют образы множеств $\mct_1$ и $\mct_2$ при любом $\lambda$, но внутри каждого из этих множеств значения $r=0$, как будет ясно ниже, не являются разделяющими за исключением конечного числа точек вида $(0,\lambda)$, в которых меняется тип абсолютного равновесия. При $r\ne 0$ типы точек $\xi_\pm(r,\lambda)$ всегда одинаковы, поэтому точка $(r,\lambda)$ является разделяющей тогда и только тогда, когда точки $\xi_\pm(r,\lambda)$ вырождены. Для исследования свойства вырожденности будем использовать пару функционально независимых (т.е. независимых почти всюду) первых интегралов $H$ (гамильтониан) и $K$ (интеграл Ковалевской\,--\,Яхья).

\begin{figure}[ht]
\centering
\includegraphics[width=8cm,keepaspectratio]{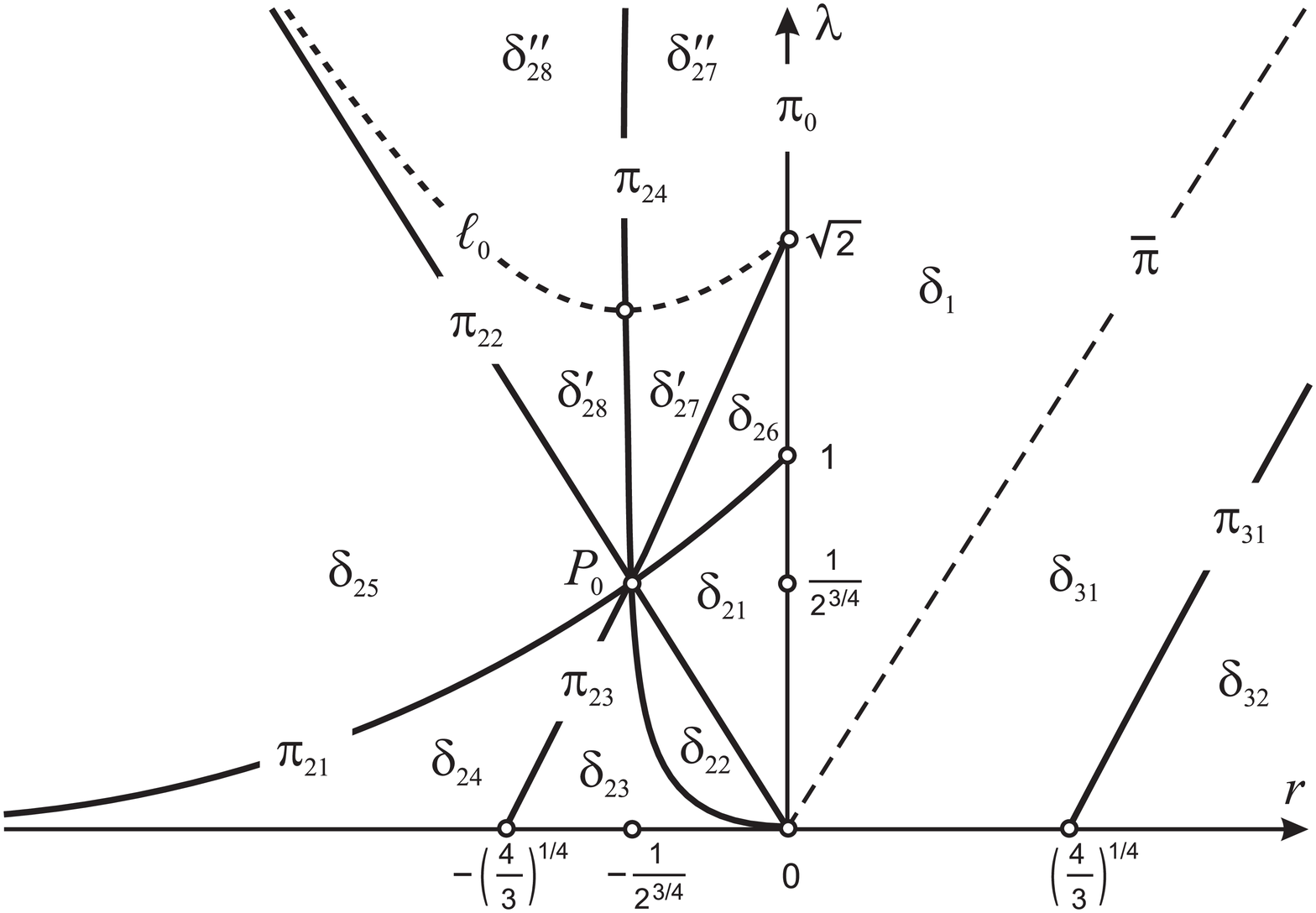}
\caption{Разделяющее множество и классы относительных равновесий.}\label{fig2}
\end{figure}

Вычислим характеристический многочлен оператора $A_H$. Для этого линеаризуем уравнения \eqref{eq:1} или, что то же самое, уравнения \eqref{eq:7} с $F=H$, после чего характеристический многочлен $6{\times}6$-матрицы линеаризованной системы в подстановке \eqref{eq:3} сократим на $\mu^2$ в соответствии с замечанием \ref{remm1}. Получим
$
\chi_H(\mu)=(\mu^2-\mu_1^2)(\mu^2-\mu_2^2),
$
где
\begin{equation}\label{eq:8}
\begin{array}{l}
  \mu_1^2  = -\ds{\frac{1}{4}}[(2r-\lambda)(r-\lambda)-d], \\
  \mu_2^2 = -\ds{\frac{1}{2(r-\lambda)}}[(2r-\lambda)(r-\lambda)r+\lambda d].
\end{array}
\end{equation}
Эти величины впервые вычислены в работе \cite{RyabUdgu}.

\begin{theorem}\label{theo1}
{\it Точки $\xi_\pm(r,\lambda)$ вырождены тогда и только тогда, когда точка $(r,\lambda) \in \mathcal{D}$ лежит на одной из кривых}
\begin{equation}\label{eq:9}
\begin{array}{l}
  \pi_{21}: \quad \ds r=\lambda-\frac{1}{\lambda^{1/3}}, \quad 0< \lambda \leqslant 1,  \\[2mm]
  \pi_{22}: \quad r=-\lambda, \quad \lambda >0,  \\
  \pi_{23}: \quad r=\ds{\frac{\sigma^4-4}{2\sigma^3}}, \qquad \lambda=\ds{\frac{3\sigma^4-4}{2\sigma^3}}, \qquad \sigma \in (\sqrt[4]{4/3},\sqrt{2}],\\
  \pi_{24}: \quad r=\ds{\frac{1}{2}}\left(\lambda-\sqrt{\lambda^2+4\lambda^{2/3}}\right), \qquad \lambda> 0,\\
  \pi_{31}: \quad r=\ds{\frac{\sigma^4-4}{2\sigma^3}}, \qquad \lambda=\ds{\frac{3\sigma^4-4}{2\sigma^3}}, \qquad \sigma \in (-\sqrt[4]{4/3},0).
\end{array}
\end{equation}
{\it Первый индекс в обозначении кривой соответствует номеру подмножества, т.е. кривая $\pi_{ij}$ содержится в соответствующем множестве $\delta_i \subset \mathcal{D}$. Только кривая $\pi_{21}$ отвечает точкам сильного вырождения.}

{\it Кривые \eqref{eq:9} порождают разбиение множества невырожденных относительных равновесий в $\Lambda({\mct})$ на $11$~классов в смысле определения~{\rm \ref{def7}} -- прообразов областей $\delta_1$, $\delta_{21}$--~$\delta_{28}$, $\delta_{31}$, $\delta_{32}$ {\rm (}рис.~{\rm \ref{fig2}}{\rm )}. Классы в прообразах подобластей $\delta_1, \delta_{21}, \delta_{26}, \delta_{27}$ связны, остальные имеют по две связных компоненты. В соответствии с обозначениями подобластей относительные равновесия имеют следующие характеристики{\rm :}}
{\it \begin{itemize}

\item $\delta_{1}, \delta_{24}, \delta_{25}, \delta_{28}, \delta_{32}$ -- $\mu_1^2<0,$ ${\mu_2^2<0}$, тип ``центр-центр'', устойчивы по всем переменным$;$

\item $\delta_{27}$ -- $\mu_1^2<0,$ ${\mu_2^2>0}$, тип ``центр-седло'', по двум переменным устойчивы, по двум -- не\-ус\-той\-чивы$;$

\item $\delta_{23}, \delta_{31}$ -- $\mu_1^2>0,$ ${\mu_2^2<0}$, тип ``седло-центр'', по двум переменным устойчивы, по двум -- не\-ус\-той\-чивы$;$

\item  $\delta_{21}, \delta_{22}, \delta_{26}$ -- $\mu_1^2>0,$ ${\mu_2^2>0}$, тип ``седло-седло'', неустойчивы по всем переменным$.$
\end{itemize}
}
\end{theorem}

Нетрудно уточнить утверждения об устойчивости, указав соответствующие направления в фазовом пространстве. Ясно, что таких уточнений требует лишь смешанный тип относительного равновесия, в котором присутствуют как ``седло'', так и ``центр''. Обозначим трехмерное многообразие, заданное уравнениями \eqref{eq:2}, сотканное из одномерных траекторий (периодических и асимптотических к относительному равновесию), через $\mathcal{M}_1$. В точке невырожденного относительного равновесия касательное пространство к $\mathcal{M}_1$ есть сумма прямой, касательной к одномерному семейству относительных равновесий, и плоскости, касательной к уровню интеграла площадей в самом $\mathcal{M}_1$, т.е. к пересечению $\mathcal{M}_1 \cap \mP^4_\ell$. Дополнение к трехмерному касательному пространству многообразия $\mathcal{M}_1$ в пятимерном фазовом пространстве $\mP^5$ есть плоскость, и эту плоскость можно выбрать так, что она окажется лежащей в трехмерном касательном пространстве к другому многообразию $\mathcal{M}_2$ частного решения, найденного в работах \cite{EIPVHDan, PVMtt71}, а именно, будет касательной плоскостью к уровню интеграла площадей в $\mathcal{M}_2$. За устойчивость относительно этой последней плоскости отвечает показатель $\mu_1^2$ (т.е. он является ``внешним'' типом относительного равновесия для инвариантного многообразия $\mathcal{M}_1$).
Направление прямой (общее для касательных пространств к $\mathcal{M}_1$ и $\mathcal{M}_2$) является собственным вектором линеаризованной системы с нулевым собственным значением и не учитывается при описании устойчивости (оно трансверсально уровню интеграла площадей). За устойчивость по отношению к направлениям в плоскости, касательной к $\mathcal{M}_1 \cap \mP^4_\ell$, отвечает показатель $\mu_2^2$.
Таким образом, при $\mu_1^2>0$, $\mu_2^2<0$ относительное равновесие устойчиво в $\mathcal{M}_1$ и неустойчиво в $\mathcal{M}_2$, при $\mu_1^2<0$, $\mu_2^2>0$, наоборот, относительное равновесие неустойчиво в $\mathcal{M}_1$ и устойчиво в $\mathcal{M}_2$.

Изображенная на рис.~\ref{fig2} кривая $\ell_0$ в области $\delta_2$ задана уравнением
\begin{equation}\notag
    \lambda (r-\lambda)+d=0 \qquad (d>0).
\end{equation}
Она отвечает случаю, когда в точках $\xi_\pm(r,\lambda)$ значение интеграла $L$ равно нулю, т.е. эти точки попадают на один интегральный уровень. Однако, как нетрудно установить из соответствующих аналитических решений для траекторий ранга $1$, две точки $\xi_\pm(r,\lambda)$ и в этом случае принадлежат разным компонентам интегрального многообразия, поэтому ни тип относительного равновесия, ни топологическая структура связной компоненты его насыщенной окрестности при пересечении кривой $\ell_0$ не изменяются. В то же время ниже будет показано различие в глобальной структуре уровня первых интегралов для подобластей в $\delta_{27},\delta_{28}$, обозначение которых снабжено штрихами.

Доказательство теоремы проведем в виде последовательности утверждений.
Обозначим для краткости $x_0=\xi_{\pm}(r,\lambda)$ и пусть
\begin{equation}\label{eq:10}
    B=\nu_1 A_H+\nu_2 A_K, \qquad \nu_1^2+\nu_2^2 \ne 0.
\end{equation}
Если у такой матрицы $B$ все собственные числа различны, то она называется \textit{регулярным элементом} (алгебры симплектических операторов \cite{BolFom}). Таким образом, условие $2^\circ$ в определении невырожденности можно назвать требованием существования регулярного элемента.

\begin{lemma}\label{lem1}
{\it Точки $x_0$ сильно вырождены тогда и только тогда, когда $(r,\lambda)\in \pi_{21}$.}
\end{lemma}
\begin{proof}
Сильно вырожденная точка отвечает существованию нулевой комбинации $B$. Располагая переменные и, соответственно, элементы матриц в порядке ${\bs \omega}, {\bs \alpha},$ приравняем к нулю элемент
\begin{equation*}
    B_{12} = \ds{\frac{1}{2}}(r-\lambda)\left[ \nu_1-2 \nu_2 \lambda (Q^2+\lambda+r)\right]=0 \quad (Q^2=\displaystyle{\frac{1}{2}\left[-r+\frac{1}{r-\lambda}d\right]} \ne 0).
\end{equation*}
Отсюда выразим $\nu_1 = 2 \nu_2 \lambda (Q^2+\lambda+r)$ ($\nu_2 \ne 0$). Тогда
\begin{equation}\notag
\begin{array}{c}
    B = \varkappa\left\|
{
\begin{array}{cccccc}
0 & 0 & 0 & 0 & 0 & 0 \\
0 & 0 & 0 & 0 & 0 & 0 \\
0 & -2\sqrt{r} & 0 & 0 & \ds{\frac{\lambda+r}{Q(\lambda-r)}} & 0 \\
0 & -2r\sqrt{r} & 0 & 0 & \ds{\frac{r(\lambda+r)}{Q(\lambda-r)}} & 0 \\
2r\sqrt{r} & 0 & -\lambda Q & -\ds{\frac{r(\lambda+r)}{Q(\lambda-r)}} & 0 & \ds{\frac{2\lambda \sqrt{r}}{\lambda-r}} \\
0 & 2 r Q & 0 & 0 & -\ds{\frac{\sqrt{r}(\lambda+r)}{\lambda-r}} & 0
\end{array}
}
    \right\|,
\end{array}
\end{equation}
где $\varkappa=2 \nu_2 (\lambda-r)Q (Q^2+\lambda)$. Очевидно, матрица в правой части ненулевая, поэтому $Q^2+\lambda=0$, что равносильно уравнению
${1+\lambda (r-\lambda)^3 =0}$ с условием ${Q^2<0}$. Но $\mathop{\rm sgn}\nolimits Q^2=\mathop{\rm sgn}\nolimits [(r-\lambda)d]$ совпадает с $\mathop{\rm sgn}\nolimits r$ (см. предложение~\ref{propos2}). Поэтому из решений уравнения нужно взять только лежащие в $\delta_2$, что и дает первую кривую \eqref{eq:9}.
\end{proof}

В частности, относительные равновесия сильно вырождены над ``узловой'' точкой
$$
P_0: \quad r=-2^{-3/4}, \qquad \lambda = 2^{-3/4},
$$
в которой пересекается пучок разделяющих кривых. Далее для краткости, говоря о слабой вырожденности над разделяющими кривыми, опускаем естественную оговорку ``за исключением точки $P_0$''.

\begin{lemma}\label{lem2}
{\it Относительные равновесия над кривой $\pi_{22}$ слабо вырождены. При этом $\det A_H \ne 0$.}
\end{lemma}
\begin{proof}
Полагая $r=-\lambda$, вычислим характеристический многочлен комбинации \eqref{eq:10}. Получим
\begin{equation*}
    \chi(\mu)=\left[\mu^2+\ds{\frac{(Z^2-2)(\nu_1 Z+\nu_2)}{2Z^3}}\right]^2, \qquad Z=\lambda^2+\sqrt{\lambda^4+1}.
\end{equation*}
Поэтому в линейной оболочке операторов $A_H ,A_K$ регулярного элемента нет, нарушено условие $2^\circ$.
Случай $Z^2=2$ приводит к точке $P_0$. В остальных точках $B\ne 0$, даже если $\nu_1 Z+\nu_2=0$ и все собственные числа равны нулю. Поэтому имеет место слабая вырожденность.

Положим $\nu_1=1, \nu_2=0$. Получим для многочлена $\chi_H(\mu)$ при $\lambda \ne 2^{-3/4}$
\begin{equation*}
    \mu_{1,2}^2=-\ds{\frac{(Z^2-2)Z}{2Z^3}} \ne 0,
\end{equation*}
т.е. второй дифференциал ограничения $H$ на $\mP^4_\ell$ в соответствующих точках невырожден.
\end{proof}

\begin{lemma}\label{lem3}
{\it Относительные равновесия над кривыми $\pi_{23}, \pi_{24}, \pi_{31}$ слабо вырождены. При этом характеристический многочлен $A_H$ имеет два нулевых собственных числа.}
\end{lemma}
\begin{proof}
Характеристический многочлен комбинации \eqref{eq:10} имеет вид
\begin{equation}\notag
    \chi(\mu)= \mu^4 -\ds{\frac{(\sigma^8+2\sigma^4-8) [-2 \nu_1 \sigma^2 + \nu_2 (\sigma^4-4 )]^2}{8 \sigma^{10}}}\mu^2
\end{equation}
при условии $(r,\lambda) \in \pi_{23} \cup \pi_{31}$. Полагая $Z=\lambda^{2/3}+\sqrt{4+\lambda^{4/3}}$, в точках кривой $\pi_{24}$ получим
\begin{equation}\notag
\begin{array}{c}
\chi(\mu)= \mu^4 + \ds{\frac{(Z^2-8) (4 + Z^2) \{
    \nu_2 [Z^2 (Z^2-8)^2-64]-4 \nu_1 Z^3\}^2}{512 Z^7}} \mu^2.
\end{array}
\end{equation}
Поэтому $\chi(\mu)$ при всех $\nu_1,\nu_2$ имеет два нулевых корня, т.е. регулярного элемента нет. В частности, два нулевых собственных числа имеет и сама матрица $A_H$ (естественно, с учетом замечания~\ref{remm1}).
\end{proof}

Итак, доказана вырожденность относительных равновесий в прообразах всех кривых, перечисленных в формулировке теоремы~\ref{theo1}. Это технически наиболее сложная часть доказательства теоремы, так как необходимо доказывать \textit{несуществование} некоторого объекта (регулярного элемента). Примеры явных доказательств вырожденности в литературе на эту тему автору неизвестны. Обычно ограничиваются утверждением об областях невырожденности (достаточными условиями). Здесь такое утверждение получить легко.

\begin{lemma}\label{lem4}
{\it За пределами кривых, перечисленных в теореме~{\rm \ref{theo1}}, относительные равновесия невырождены. При этом во всех таких точках второй дифференциал приведенного гамильтониана невырожден.}
\end{lemma}
\begin{proof}
По лемме~\ref{lem1} сильного вырождения в этой области быть не может. В качестве возможного регулярного элемента возьмем саму матрицу $A_H$. Она может иметь совпадающие собственные числа лишь в трех случаях $\mu_1^2=0$, $\mu_2^2=0$ или $\mu_1^2=\mu_2^2$.
Первое условие при $r\ne 0$ дает уравнение
$(r-\lambda)(3r-\lambda)=4$ с условием $\mathop{\rm sgn}\nolimits (2r-\lambda)=\mathop{\rm sgn}\nolimits r$.
Отсюда подстановкой $r=\lambda -\sigma$ получаем параметрические выражения кривых $\pi_{23}, \pi_{31}$. При $r=0$ из $\mu_1^2=0$ следует $\lambda=\sqrt{2}$, $d=2>0$. Это -- граничная точка кривой $\pi_{23}$.
Случай $\mu_2^2=0$ сводится к уравнению $r(r-\lambda)=\lambda^{2/3}$,
которое при $\lambda>0$ определяет ровно по одной точке в $\delta_2$ ($r<0$) и в $\delta_3$ ($r>\lambda$). В этих областях следует считать $d>0$. Но тогда в области $\delta_3$ будет
${(2r-\lambda)(r-\lambda)r=-\lambda r <0}$, что не так. Поэтому допустимым решением здесь является лишь кривая $\pi_{24}$.

Вне кривых $\pi_{23}, \pi_{31}, \pi_{24}$ многочлен $\chi_H(\mu)$ не имеет нулевых корней, что, в частности, означает невырожденность ограничения второго дифференциала функции $H$ на фазовое пространство $\mP^4_\ell$ соответствующей приведенной системы.

Условие $\mu_1^2=\mu_2^2$ за пределами кривой $\pi_{22}$ дает
$$
   (2r-\lambda)(r-\lambda)+d=0.
$$
Выполняя подстановку $r=\lambda-\sigma$, получим то же параметрическое представление, что и на кривых $\pi_{23}, \pi_{31}$, но при $\sigma \geqslant \sqrt{2}$ и $d<0$. Характеристический многочлен $\chi_H(\mu)$ по условию имеет хотя и ненулевые, но попарно совпадающие собственные числа и
задачу о невырожденности не решает. Вычислим, однако, характеристический многочлен $A_K$:
$$
\chi_K(\mu)=\left[\mu^2 + \ds{\frac{(\sigma^4-4)^2 (4 + \sigma^4)}{\sigma^{14}}}\right] \left[\mu^2 + \ds{\frac{(\sigma^4+4) (3 \sigma^8 - 7 \sigma^4 + 4)^2}{\sigma^{14}}}\right].
$$
При $\sigma > \sqrt{2}$ все его корни различны, поэтому он представляет собой искомый регулярный элемент.
В граничной точке $\sigma =\sqrt{2}$, которая соответствует нижнему абсолютному равновесию тела $c_+$ при $\lambda=\sqrt{2}$ ($r=0$, $d = -2$), многочлен
$$
\chi_K(\mu)=(\mu^2+36)\mu^2
$$
регулярным элементом не является. Однако для комбинации \eqref{eq:10} имеем $\chi(\mu)=(\mu^2+\nu_1^2) \left[ \mu^2+(\nu_1-6\nu_2)^2\right].$ Поэтому при любых $\nu_1 \ne 0$, $\nu_2 \ne 0$ матрица $B$ есть регулярный элемент, и соответствующие точки невырождены.
\end{proof}

Доказательство теоремы завершается применением лемм~\ref{lem1}\,--\,\ref{lem4} и непосредственным определением знаков величин \eqref{eq:8} в подобластях множества $\mathcal{D}$.

\section{Топология интегральных уровней относительных равновесий}
Пусть $x$ -- точка невырожденного относительного равновесия, $\ell=L(x)$. Рассмотрим $\mathcal{F}_\ell=H{\times}K: \mP^4_\ell \to \bbR^2$ -- интегральное отображение приведенной системы. Обозначим через $J(x)$ полный прообраз точки $\mathcal{F}_\ell(x)$ -- критическую интегральную поверхность, а через $U(x)$ -- достаточно малую насыщенную окрестность поверхности $J(x)$, не содержащую относительных равновесий с другими значениями отображения $\mathcal{F}_\ell$. Поверхность $J(x)$ может состоять из нескольких связных компонент. Как следует из формул \eqref{eq:3}, компонента точки $x$ всегда содержит ровно одно относительное равновесие. Одновременно могут существовать компоненты, содержащие ровно одну критическую окружность. Количество и топология этих компонент устанавливаются по сводке результатов для критических подсистем, приведенных в работе \cite{KhRyabUdgu}. Рассматривая точку $x$ в каждой из двух содержащих ее критических подсистем \cite{KhRyabUdgu}, анализируем информацию по прилегающим  областям в образе критической подсистемы. Эти области порождают дуги бифуркационной диаграммы отображения $\mathcal{F}_\ell$ в окрестности точки $x$. Структура критического множества в прообразах этих дуг и перестройки в $U(x)$ при их пересечении находятся по таблицам из \cite{KhRyabUdgu}. После этого тип круговых молекул самих относительных равновесий и лежащих на том же уровне критических периодических траекторий вместе с метками однозначно устанавливается исходя из исчерпывающего описания круговых молекул невырожденных особенностей низкой сложности \cite{BolFom}. Кроме компонент, содержащих критические точки, в $J(x)$ могут входить и регулярные торы, заполненные двояко-периодическими траекториями. Их количество однозначно устанавливается по виду бифуркационной диаграммы, дополненной указанием атомов на дугах.

\def\wid{18mm}

\begin{table}[h!]

{\small

\renewcommand{\arraystretch}{0}

\begin{center}

\begin{tabular}{|m{10mm}|m{16mm}|m{17mm}|m{24mm}|m{24mm}|m{8mm}|}
\multicolumn{6}{r}{Таблица}\\[2mm]
\hline
\hspace*{-2mm}{\renewcommand{\arraystretch}{0.8}\fns{\begin{tabular}{c} Класс\\точек\end{tabular}} }
&
\hspace*{-2.5mm}{ \renewcommand{\arraystretch}{0.8}\fns{\begin{tabular}{c} \ru{10} Компо-\\
нент\\ в прообразе
\end{tabular}}  }
&
\hspace*{-3mm}{ \fns{\renewcommand{\arraystretch}{0.8}\begin{tabular}{c} \ru{10} Особые\\
траектории\\ в прообразе
\end{tabular}}  } & \fns{\renewcommand{\arraystretch}{0.8}\begin{tabular}{c}Диаграмма\end{tabular}} &
{\renewcommand{\arraystretch}{0.8} \fns{\begin{tabular}{c} Молекулы\\в прообразе\end{tabular}} }
&
{
\hspace*{-3mm}\fns{\renewcommand{\arraystretch}{0.8}\begin{tabular}{c}
 Регул.\\ торы
\end{tabular}} }\\
\hline
\hf{$\delta_1$} & \hf{1} & \hf{$ {p_{CC}} $} & \hf{\includegraphics[width=\wid,keepaspectratio]{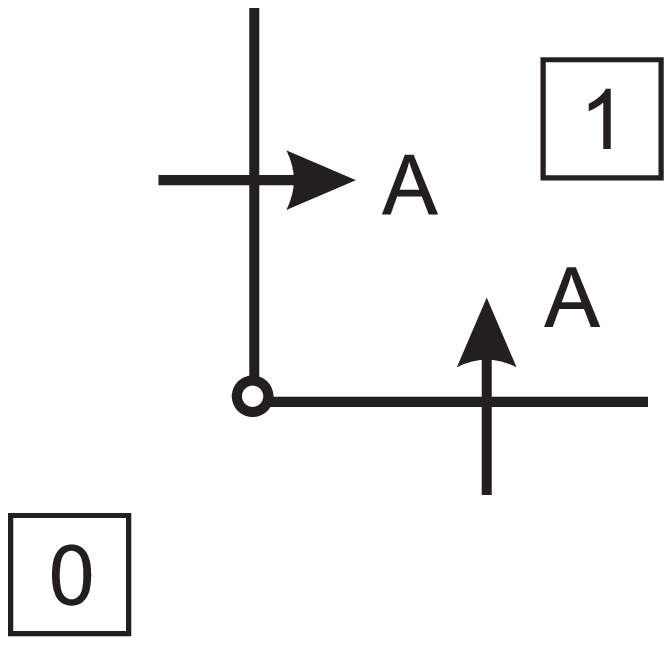}} & \hf{\includegraphics[width=\wid,keepaspectratio]{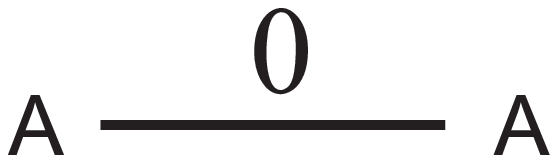}} & \hf{0}\\
\hline
 \hf{$\delta_{25}$} & \hf{3} & \hf{${p_{CC}} \cup S_E^1$}& \hf{\includegraphics[width=\wid,keepaspectratio]{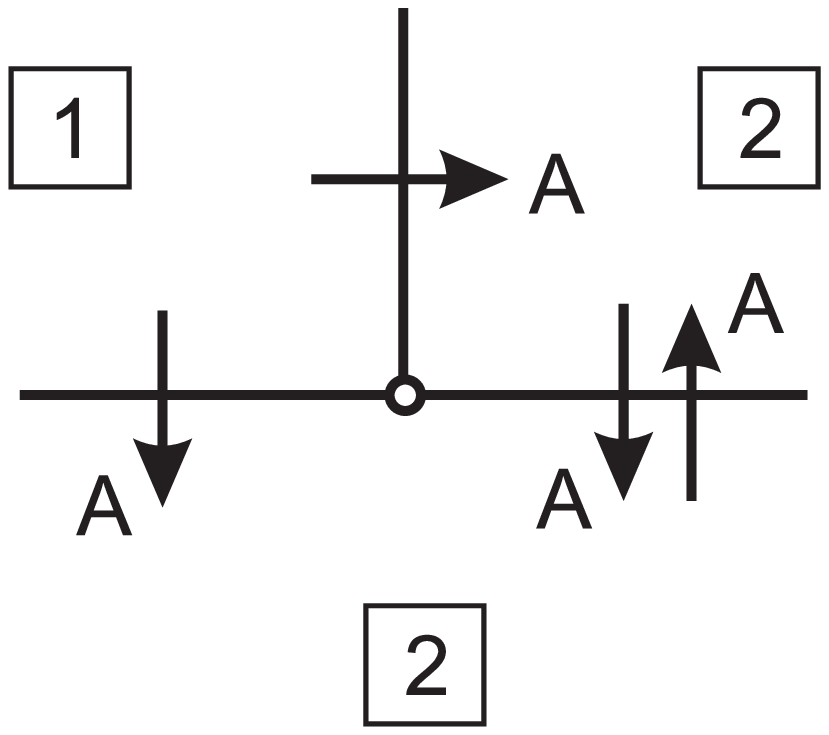}} & \hf{\includegraphics[width=\wid,keepaspectratio]{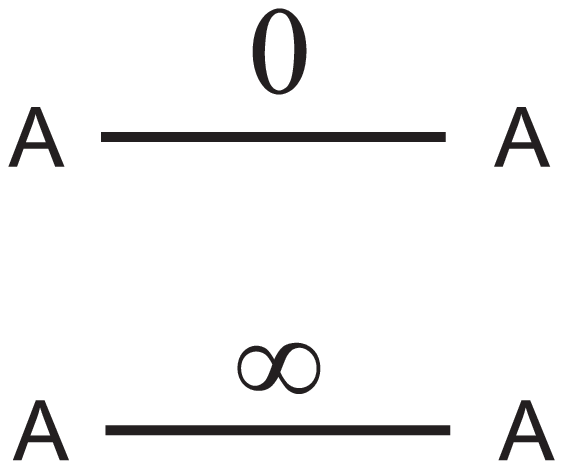}} & \hf{1} \\
\hline
\hf{$\delta^{\prime\prime}_{28}, \delta_{32}$} & \hf{2} & \hf{${p_{CC}} \cup S_E^1$}& \hf{\includegraphics[width=\wid,keepaspectratio]{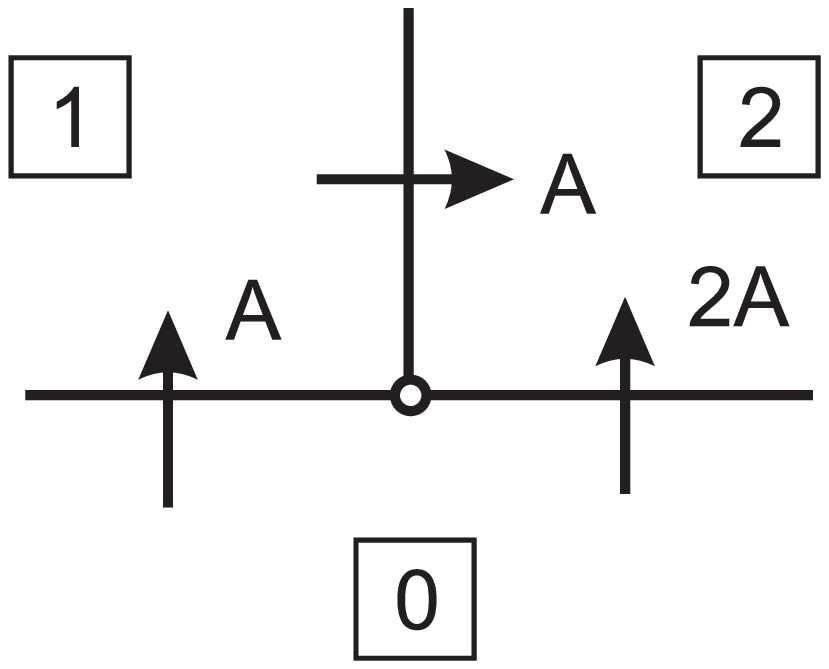}} & \hf{\includegraphics[width=\wid,keepaspectratio]{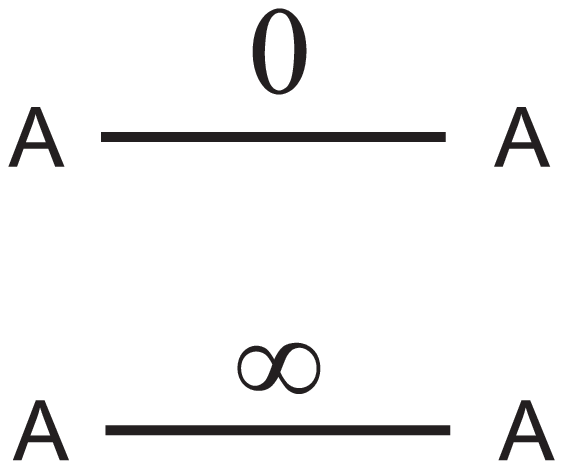}} & \hf{0} \\
\hline
\hf{$\delta_{24},\delta^{\prime}_{28}$} & \hf{4} & \hf{${p_{CC}}\cup 2S_E^1$} & \hf{\includegraphics[width=\wid,keepaspectratio]{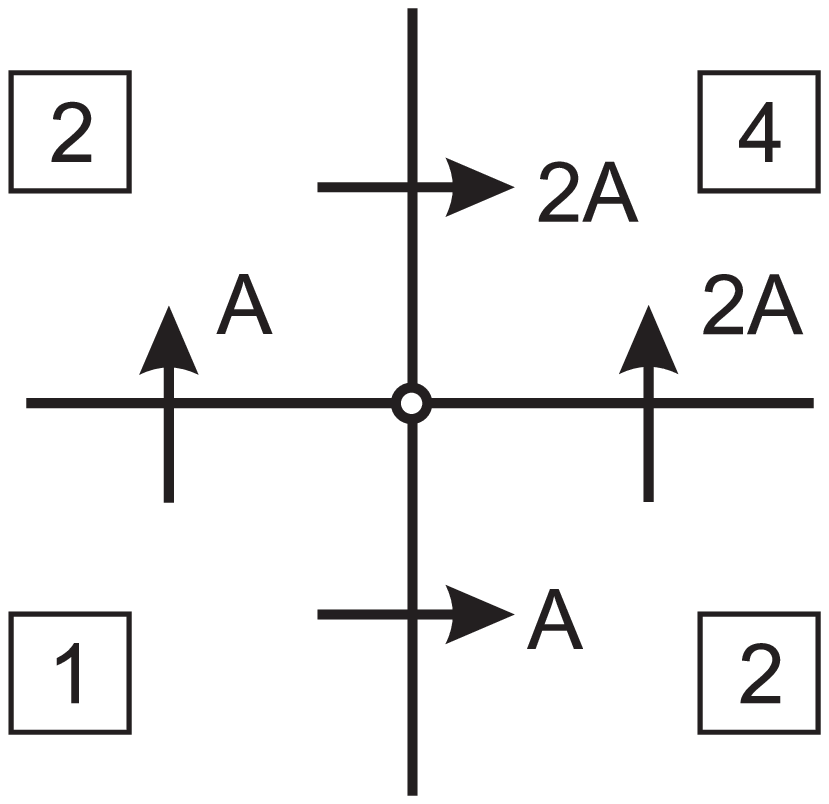}} & \hf{\includegraphics[width=\wid,keepaspectratio]{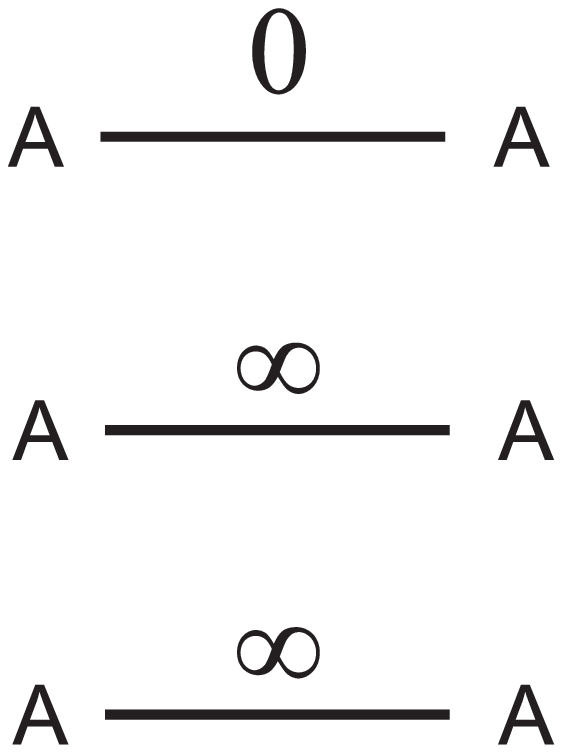}} & \hf{1}\\
\hline
\hf{$\delta_{23},\delta^{\prime}_{27}$} & \hf{2} & \hf{${p_{CS}}\cup S_H^1$} & \hf{\includegraphics[width=\wid,keepaspectratio]{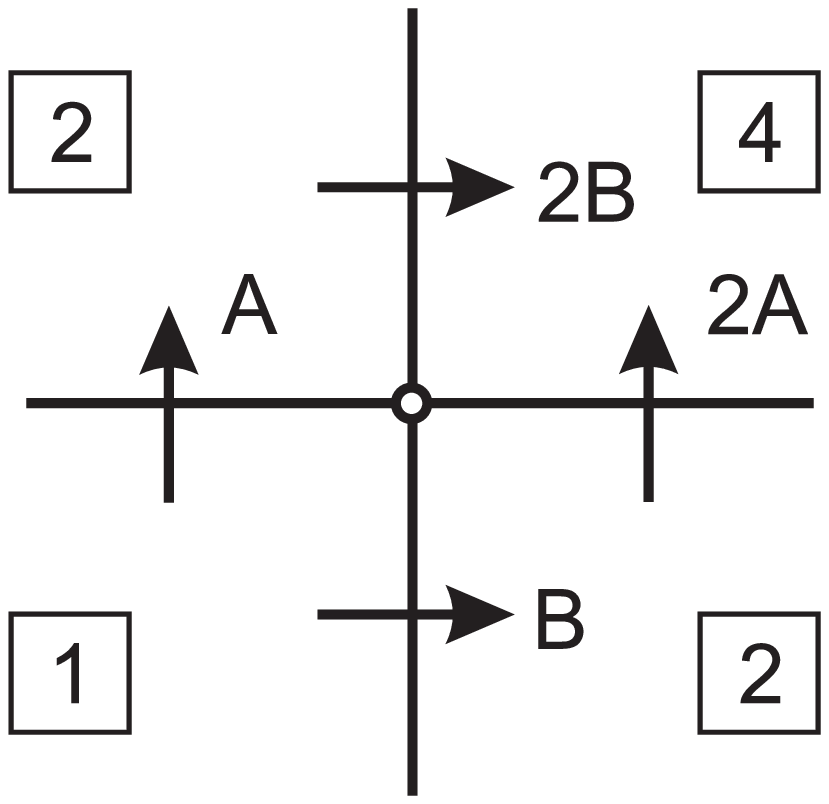}} & \hf{\includegraphics[width=\wid,keepaspectratio]{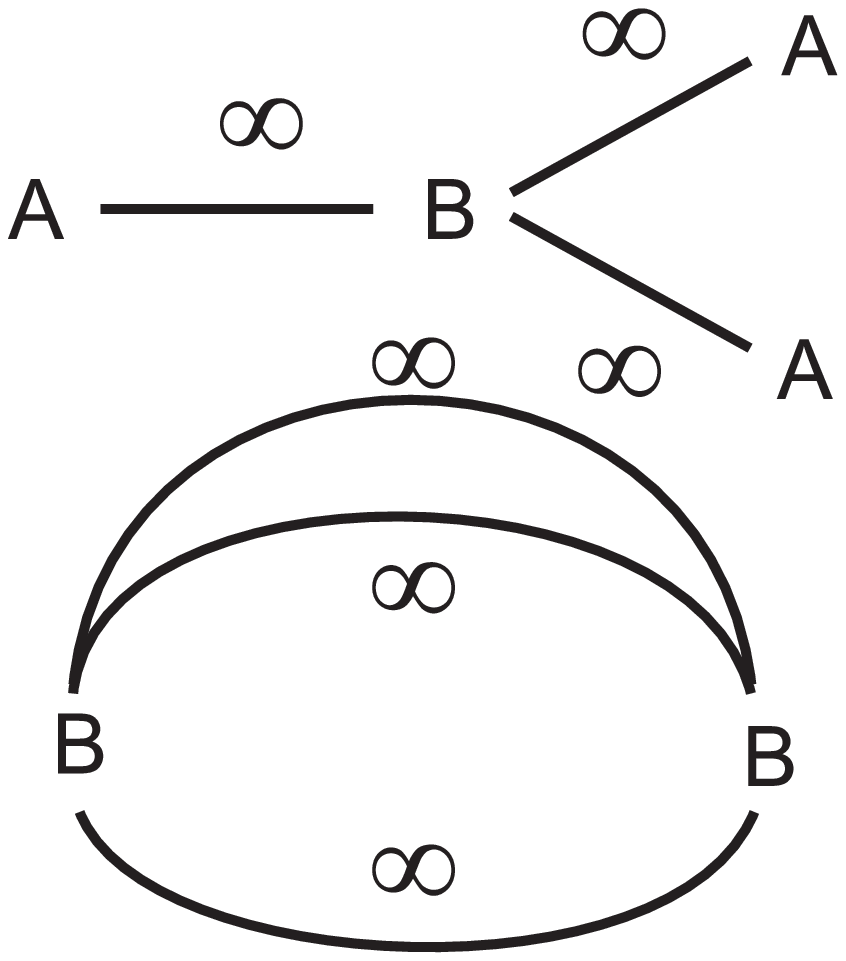}} & \hf{0}\\
\hline
\hf{$\delta_{31},\delta^{\prime\prime}_{27}$} & \hf{1} & \hf{${p_{CS}}$} & \hf{\includegraphics[width=\wid,keepaspectratio]{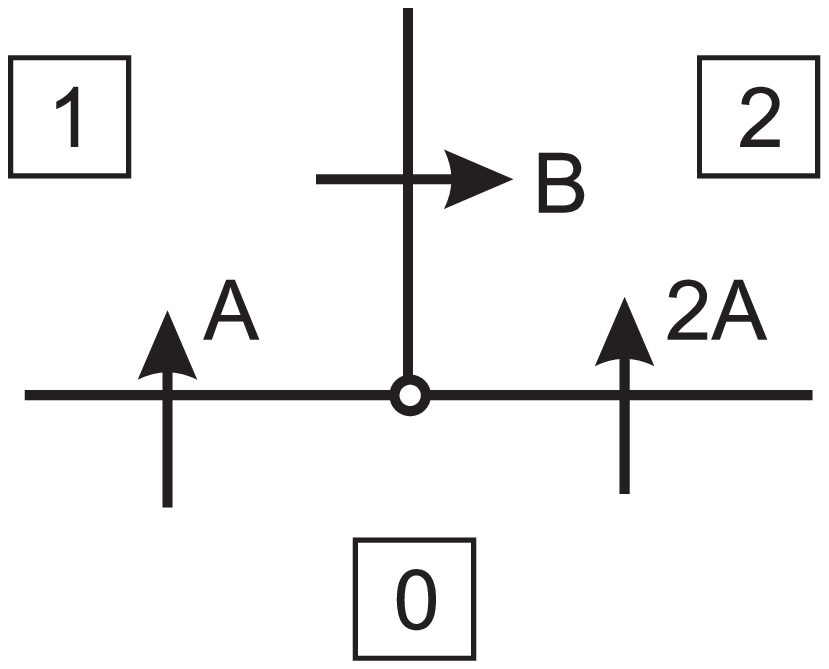}} & \hf{\includegraphics[width=\wid,keepaspectratio]{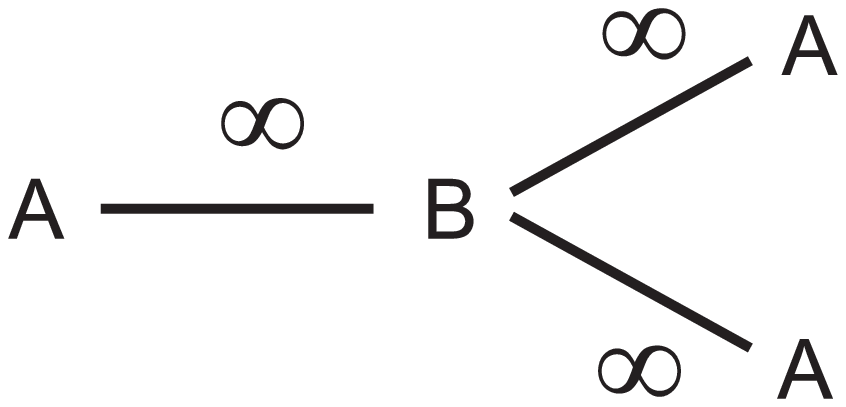}} & \hf{0}\\
\hline
\hf{$\delta_{22},\delta_{26}$} & \hf{1} & \hf{${p_{SS}}$} & \hf{\includegraphics[width=\wid,keepaspectratio]{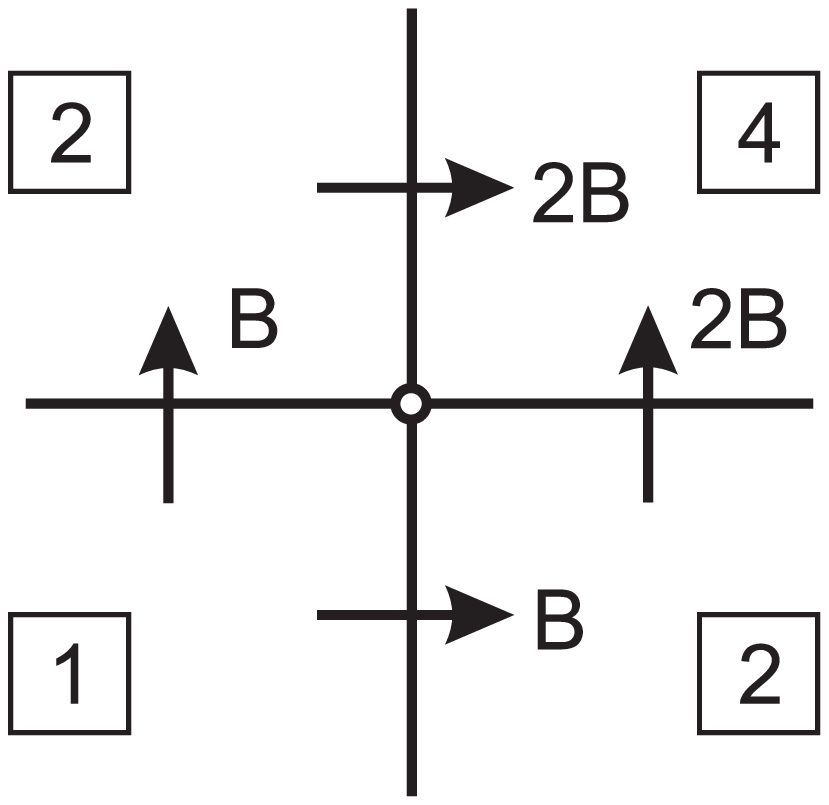}} &
\hf{\includegraphics[width=\wid,keepaspectratio]{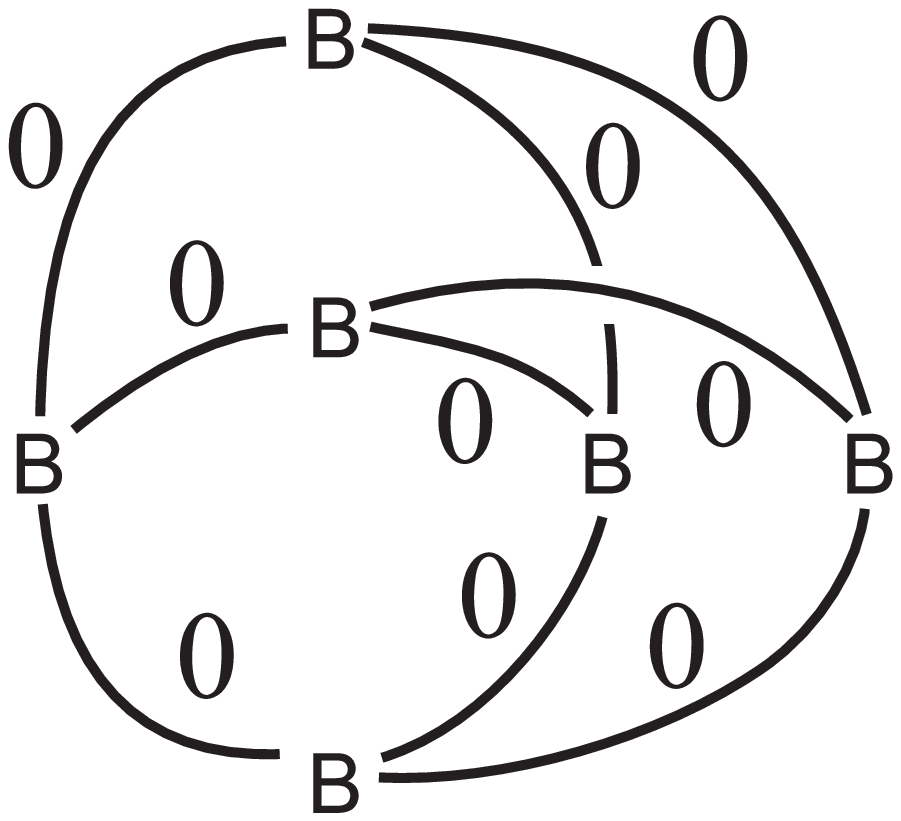}} & \hf{0}\\
\hline

\hf{$\delta_{21}$} & \hf{1} & \hf{${p_{SS}}$} & \hf{\includegraphics[width=\wid,keepaspectratio]{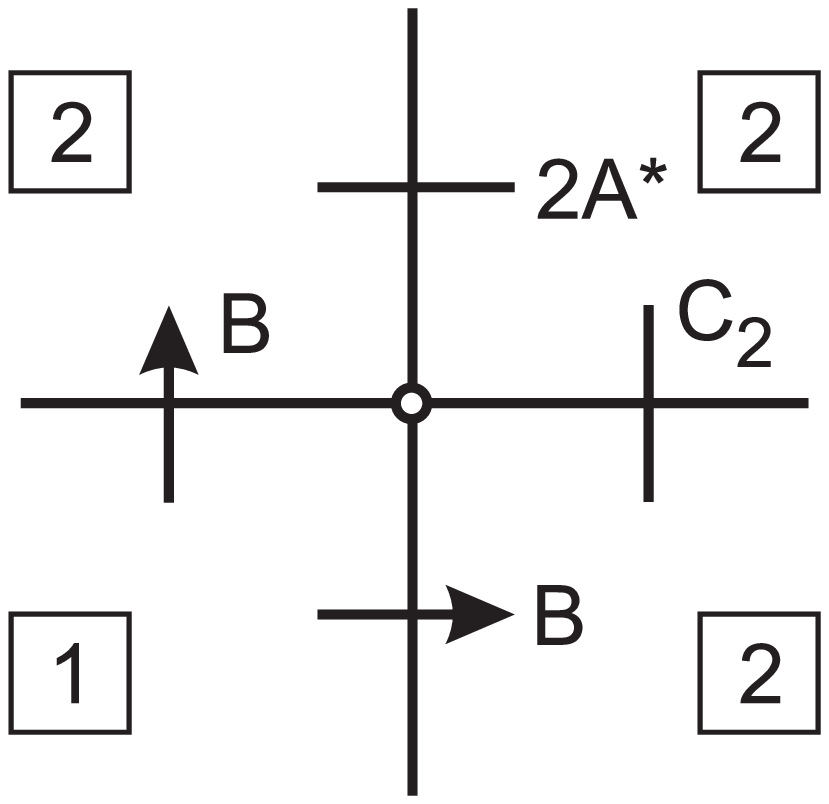}} &
\hf{\includegraphics[width=\wid,keepaspectratio]{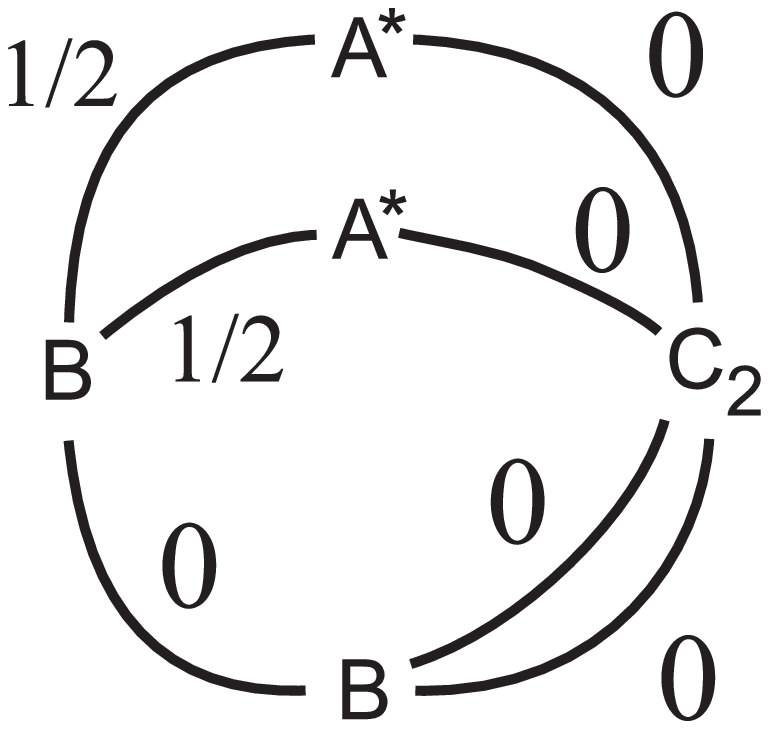}} & \hf{0}\\
\hline
\end{tabular}
\end{center}

}
\end{table}

%\cleapage

Результат топологической классификации приведен в таблице. Для круговых молекул указаны только $r$-метки с единственной целью -- отличить молекулы относительных равновесий типа ``центр-центр'' (метка $r=0$) от молекул лежащих на том же уровне эллиптических периодических траекторий (метка $r=\infty$). В действительности же, здесь все метки (включая $\varepsilon, n$-метки) выставляются автоматически в соответствии с \cite{BolFom}.

На фрагментах диаграммы в окрестности относительного равновесия указаны атомы, возникающие на критических окружностях при пересечении дуг диаграммы. Здесь встречаются лишь атомы типов $A,B,A^*,C_2$. Для несимметричных атомов стрелкой указано направление возрастания числа торов. Само это число в регулярных областях указано в рамке. При описании особых траекторий (сингулярной компоненты интегрального многообразия) ${p_{CC}}$ -- это неподвижная точка типа ``центр-центр'', ${p_{CS}}$~-- критическое многообразие неподвижной точки типа ``центр-седло'' -- восьмерка, ${p_{SS}}$~-- критическое многообразие неподвижной точки типа ``седло-седло'' -- две восьмерки с общей центральной точкой и приклеенные к ним четыре прямоугольника, заполненных асимптотическими траекториями из регулярных точек (правило склейки полностью определено соответствующей круговой молекулой). Через $S^1_E$ обозначена периодическая траектория эллиптического типа, исчерпывающая соответствующую связную компоненту, а через $S^1_H$ -- поверхность периодической траектории гиперболического типа, отвечающая атому типа $B$ (прямое произведение восьмерки на окружность). Полученная классификация по количеству классов и по виду круговых молекул полного прообраза значения отображения момента в относительных равновесиях отличается от результатов, представленных недавно в \cite{Log}. Ввиду отсутствия в цитируемой работе точного определения принципа классификации детальное сопоставление результатов не проводилось.

\section*{Заключение}
Сформулируем кратко результаты классификации относительных равновесий случая Ковалевской\,--\,Яхья.

1) В фазовом пространстве $\mP^5=\bbR^3{\times}S^2$ при любом $\lambda>0$ множество относительных равновесий $\mct(\lambda)$ однопараметрическое, имеет четыре связных компоненты. Две из них сохраняются симметрией ${\bs \tau}$, меняющей знак интеграла площадей, две остальных симметричны друг другу.

2) Объявляя эквивалентными в расширенном множестве $\cup_\lambda (\mct(\lambda),\lambda)$ относительные равновесия, которые можно перевести друг в друга непрерывным изменением параметров или симметрией ${\bs \tau}$ с сохранением топологии связной насыщенной окрестности, получим 11 классов эквивалентности. Разделяющие кривые в пространстве параметров и типы относительных равновесий в классах определены в теореме~\ref{theo1}. Для связных окрестностей относительных равновесий имеется четыре вида круговых молекул.

3) Требуя сохранение при непрерывном изменении параметров полного уровня первых интегралов, отвечающего относительному равновесию, приходим к 13 классам. Для таких интегральных многообразий получаем семь различных сочетаний круговых молекул на одном уровне без учета наличия регулярных торов, и восемь -- с учетом этого наличия (два последних столбца таблицы).

%\clearpage

%\clearpage

\begin{abstract}[en]
\title{Analytical classification of the permanent rotations\\of the Kowalevski\,--\,Yehia gyrostat}
\author{M.P.~Kharlamov}
The complete investigation of the permanent rotations of a gyrostat in the integrable case of Kowalevski\,--\,Yehia is presented.
The notion of equivalence classes is given with respect to the defining parameters, the separating set is constructed.
For each class the type of a singularity is calculated as the type of a fixed point in the reduced system. The detailed character of stability is obtained, and the structure of local Liouville foliation is shown.
\end{abstract}

\begin{abstract}[uk]
\title{Аналiтична класифiкация рівномірних обертань\\гiростата Ковалевско\"{\i}\,--\,Яхья}
\author{М.П.~Харламов}
Представлено повне дослідження безлічі рівномірних обертань гіростата в разі інтегрованості Ковалевсько\"{\i}\,--\,Яхья. Введено поняття класів еквівалентності відносно визначальних параметрів, побудована розділяюча безліч. Для кожного класу обчислений тип особливості як тип нерухомо\"{\i} точкi в приведеній системі, отриманий детальний характер стійкості, вказана структура локального шарування Ліувіля.
\end{abstract}

\makeaddressline

\end{document}